\documentclass[12pt]{article}

\usepackage{graphicx}
\usepackage[dvips]{color}

\usepackage{amsmath,amssymb,epsf,epic}

\newcommand{\un}{\mathbb I}
\newcommand{\ra}{\rightarrow}

\newcommand{\bra}{\langle} 

\newcommand{\ket}{\rangle}

\newcommand{\D}{{\mathbb D}}

\renewcommand{\vec}[1]{\mathbf{#1}}

\newcommand{\be}{\begin{equation}}
\newcommand{\ee}{\end{equation}}
\newcommand{\bea}{\begin{eqnarray}}
\newcommand{\eea}{\end{eqnarray}}

\newcommand{\circa}{{\circlearrowleft}}

\newcommand{\eps}{\epsilon}

\newcommand{\ffi}{\varphi}

\newcommand{\ep}{\hfill  {\vrule height 10pt width 8pt depth 0pt}}

\newcommand{\grintl}{[\kern-.18em [}
\newcommand{\grintr}{]\kern-.18em ]}

\newcounter{resultcounter}[section]

\newtheorem{thm}[resultcounter]{Theorem}
\newtheorem{lem}[resultcounter]{Lemma}
\newtheorem{prop}[resultcounter]{Proposition}
\newtheorem{cor}[resultcounter]{Corollary}
\newtheorem{definition}[resultcounter]{Definition}
\newtheorem{rem}[resultcounter]{Remark}
\newtheorem{rems}[resultcounter]{Remarks}

  \def\cC{{\cal C}}
  
 \def\cH{{\cal H}} 
 \def\cK{{\cal K}}

 \def\cT{{\cal T}}

\newcommand{\R}{{\mathbb R}}
\newcommand{\N}{{\mathbb N}}

\newcommand{\C}{{\mathbb C}}
\newcommand{\Z}{{\mathbb Z}}

\newcommand{\I}{{\mathbb I}}
\newcommand{\T}{{\mathbb T}}
\newcommand{\Ss}{{\mathbb S}}

\def\fS{{\mathfrak S}}
\def\fB{{\mathfrak B}}

\def\proof{\noindent{\bf Proof:}\ \ }

\begin{document}
\title{Spectral Properties of Quantum Walks \\ on Rooted Binary Trees}
 \author{ Alain Joye\footnote{ UJF-Grenoble 1, CNRS Institut Fourier UMR 5582, Grenoble, 38402, France} \footnote{Partially supported by the Agence Nationale de la Recherche, grant ANR-09-BLAN-0098-01} \and Laurent Marin\footnotemark[1] }

\date{ }

\maketitle
\vspace{-1cm}

\thispagestyle{empty}
\setcounter{page}{1}
\setcounter{section}{1}

\vspace{.6cm}

\centerline{\em  Dedicated to Herbert Spohn in celebration of his 65th birthday}

\vspace{.2cm}

\setcounter{section}{0}

\abstract{
We define coined Quantum Walks on the infinite rooted binary tree given by unitary operators $U(C)$ on an associated infinite dimensional Hilbert space, depending on a unitary coin matrix $C\in U(3)$, and study their spectral properties. For circulant unitary coin matrices $C$, we derive an equation for the Carath\'eodory function associated to the spectral measure of a cyclic vector for $U(C)$. This allows us to show that for all circulant unitary coin matrices, the spectrum of the Quantum Walk has no singular continuous component. Furthermore, for coin matrices $C$ which are orthogonal circulant matrices, we show that  the spectrum of the Quantum Walk  is absolutely continuous, except for four coin matrices for which the spectrum of $U(C)$ is pure point.
}

\thispagestyle{empty}
\setcounter{page}{1}
\setcounter{section}{1}

\setcounter{section}{0}

\section{Introduction}

Simple or coined Quantum Walks are defined as the discrete dynamics of a particle with an internal degree of freedom, called coin state, on a graph. The dynamics of a simple Quantum Walk, QW for short, consist in the repeated action of the composition of a unitary coin matrix on the internal degree of freedom followed by a finite range shift on the graph, conditioned on the coin state. In other words, the shift makes the particle propagate on the graph, whereas the coin state somehow selects the direction of the motion. QWs of this sort are sometimes considered as quantum analogs of classical random walks on the underlying graph, see {\it e.g.} \cite{ADZ, Ke, Ko, V-A}. 
Consequently, QWs play an important role in computer science, in particular in the development of search algorithms by the quantum computing community, see \cite{AAKV, S, MNRS}. QWs have also been introduced to provide effective dynamics of physical quantum systems in certain asymptotic regimes. For example, quantum lattice gases, or the dynamics of an electron in a two dimensional random background potential submitted to a large perpendicular magnetic field can be described in terms of QWs. Also, the dynamics of atoms trapped in time dependent optical lattices, that of ions caught in suitably tuned magnetic Paul traps or the propagation of polarized photons in networks of waveguides  are experimentally well captured by deterministic or random QWs {\it e.g.} \cite{M, CC, Ketal, sciarrino, Zetal}. These models naturally led to the study of random QWs, \cite{KLMW, ABJ2, J4, J5, HJ2}. Further generalizations have been proposed, notably extensions from the unitary framework to completely positive maps \cite{AAKV, Gu, APSS} defining open QWs, and from stationary to time-dependent QWs \cite{AVWW, J3, HJ}. 
The popularity of QWs in different fields illustrated by this non exhaustive list is certainly due to their flexibility in modeling and to the tractable, yet non trivial, mathematical analyses of their transport and spectral properties that their structure allows.
\medskip

In this paper, we construct QWs on rooted binary trees and we address their spectral properties in a stationary, deterministic and homogeneous setup:

On the one hand, we provide a general definition of  coined QWs on the infinite rooted binary tree, which, despite the interest of this particular infinite graph for many applications, does not seem to appear in the literature. Informally, the QW describes the dynamics of a particle with coin state in $\C^3$ on the binary tree, with certain boundary conditions at the root. The shift makes the particle jump from a site of the rooted binary tree to the nearest neighbors of this site, and the coin state update is performed by a matrix $C\in U(3)$.  The resulting QW is denoted here by the unitary operator $U(C)$ on the associated Hilbert space, where the coin matrix $C$ appears as a parameter. The precise construction of $U(C)$ is provided in Section 2, see Definition \ref{bc1}, together with some of its symmetry properties. 

On the other hand, we address the spectral properties of the QW when the coin matrix $C$ belongs to $U(3)\cap \mbox{Circ} (3)$, the set of circulant unitary matrices introduced in Section 3. We make use of the structure of the tree to derive an equation for the Carath\'eodory function of the spectral measure of a cyclic vector for $U(C)$ in Section 4, Theorem \ref{maintech}. This technical result allows us discuss some spectral features of the corresponding QWs, as stated in Corollary \ref{critacucirc} and Theorem \ref{mainthm}. In particular, we prove the absence of singular continuous spectrum for $U(C)$, when $C\in U(3)\cap \mbox{Circ} (3)$. Moreover, when $C\in O(3)\cap \mbox{Circ} (3)$, the set of orthogonal circulant matrices, we prove that $U(C)$ has purely absolutely continuous spectrum, unless $C$ is a permutation matrix, up to the sign, distinct from the identity. In the latter cases the spectrum of $U(C)$ consists in six infinitely degenerate eigenvalues. 

Let us note that there exist constructions of QWs on binary trees, \cite{Detal}, which, however, do not have the simple structure of coined QWs, and therefore lack the quantum mechanical interpretation of discrete dynamics of a particle with internal degree of freedom (or spin). Our definition is based on \cite{HJ2} which considers random coined QWs on the full {homogeneous tree} of coordination number 3, with different families of boundary conditions that are discussed below. An analogous special symmetric QW on the homogeneous tree was constructed \cite{CHKS}, with the main difference that the repeated action of the coin state conditioned shift alone does not induce propagation on the tree. 

While the results of \cite{HJ2} concern the localization-delocalization spectral transition in a random framework, they don't address the deterministic homogeneous QW on the tree. In the analogous self-adjoint setup where the operator of interest is the Laplacian on the tree, the spectral analysis is performed either by explicit diagonalization via the Fourier-Helgason transform, by reduction to an infinite direct  sum  of one-dimensional Jacobi matrices, or by computation of the resolvent operator, making use of the symmetries of the Laplacian, see {\em e.g.} \cite{CdVT}, \cite{Si}. We follow the last mentioned route in our analysis of $U(C)$. Even though we cannot prove it, we expect the spectrum of $U(C)$ to be absolutely continuous for any $C\in U(3)$ that is not a permutation matrix, up to phases, distinct from the identity.

{\bf Acknowledgements} A.J. wishes to thank Eman Hamza for many useful discussions at an early stage of this project.

\section{Quantum Walks on the Binary Tree}\label{setup}

We recall here the part of the general framework developed in \cite{HJ2} which is relevant for the definition of QWs on binary trees. We start with the definition of QWs on the homogeneous tree $\cT_3$, of coordination number equal to 3.

\subsection{Homogeneous Tree}

Let $\cT_3$ be the tree corresponding to the free group generated by 
\be
A_3=\{a,b,c\} \ \ \ \mbox{with}\ \ \ a^2=b^2=c^2=e, \ \  \ \ \mbox{$e$ the neutral element.}
\ee
Choose a vertex of $\cT_3$ to be the root of the tree, denoted by $e$. Each vertex $x=x_{1}x_{2}\dots x_{n}$, $n\in\N$ of $\cT_3$ is a reduced word made of finitely many letters from the alphabet $A_3$. Accordingly, an edge of $\cT_3$ consists in a pair of vertices $(x,y)$ such that $xy^{-1}\in A_3$. This last relation defines nearest neighbors in $\cT_3$ and any vertex has thus $3$ nearest neighbors. Any pair of vertices $x$ and $y$ can be joined by a unique set of edges, or path of $\cT_3$. The distance $|x|$ of a vertex $x=x_{1}x_{2}\dots x_{n}$ to the root  is $n$  and we denote by $d(x,y)$ the distance between two arbitrary vertices.
Given the order $A_3=\{a, b, c\}$, the sequence $xa$, $xb$, $xc$ of nearest neighbors of  any $x$, is ordered around $x$ in the positive orientation. By iteration, this provides a unique numbering of the vertices of $\cT_3$, which we identity with $\cT_3$, see  Figure \ref{T3}. 

\begin{figure}[htbp]
   \begin{center}
      \includegraphics[scale=.3]{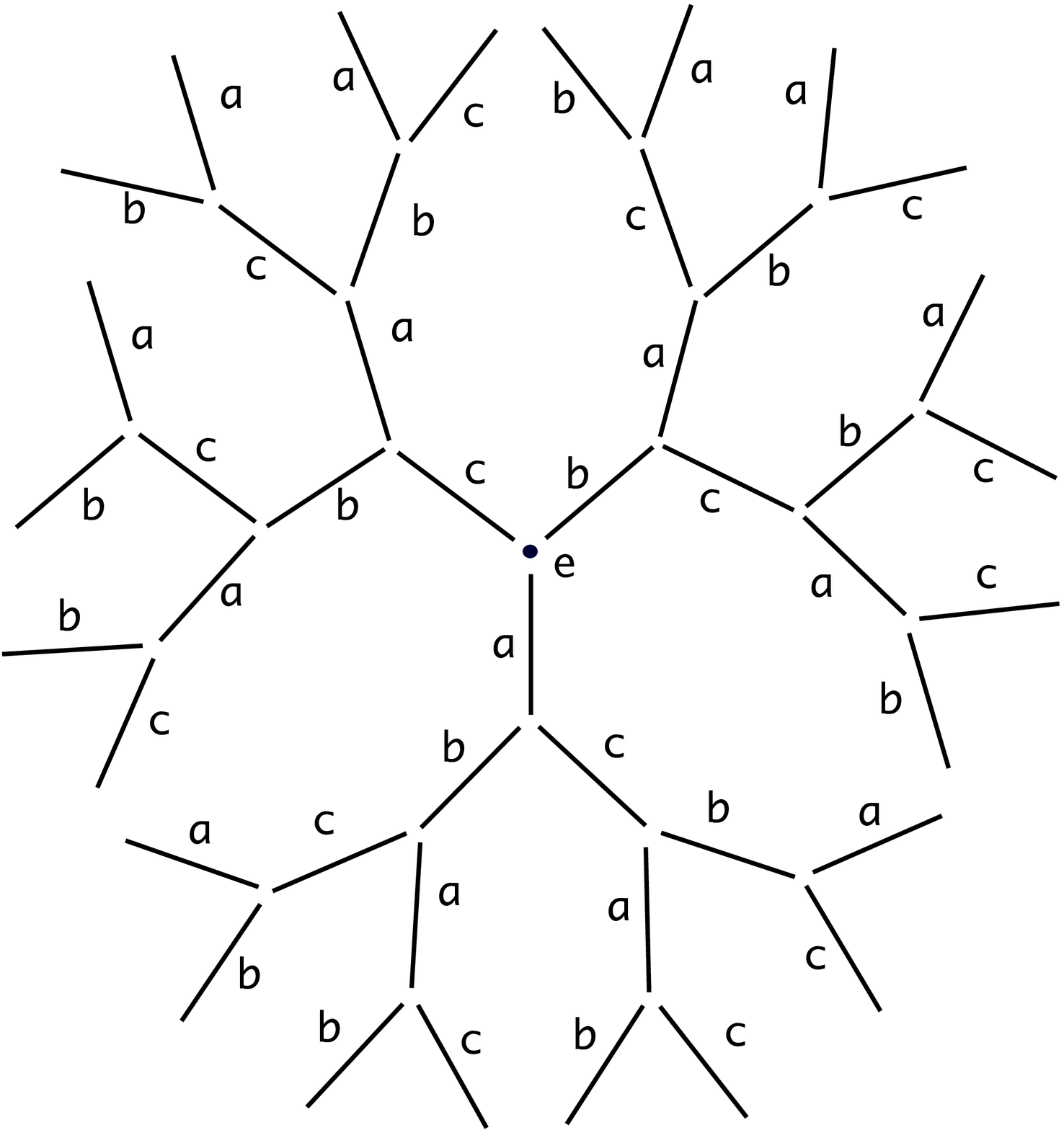}
   \end{center}
   \caption{\footnotesize Construction of $\cT_3$.}\label{T3}
\end{figure}

The Hilbert space $\cK_3$ of the QW on the homogeneous ternary tree $\cT_3$ consists in two parts. The configuration part of the Hilbert space of the QW  is defined by
\be
l^2(\cT_3)=\Big\{\psi=\sum_{x\in \cT_3}\psi_x |x\ket \ \mbox{s.t.} \  \psi_x\in \C,  \ \sum_{x\in \cT_3}|\psi_x|^2<\infty\Big\},
\ee
where $|x\ket$ denotes the element of the canonical basis of $l^2(\cT_3)$ sitting at vertex $x$. The internal degree of freedom of the quantum walker, aka spin or coin state, lives in $\C^3$, the coin Hilbert space, so that the total Hilbert space of the walker is 
\be\label{K3}
\cK_3=\cT_3\otimes \C^3.
\ee 
The canonical basis of the coin Hilbert space is labelled by the same symbols and is thus given by the ordered set $\{|a\ket, |b\ket,|c\ket\}$. Finally,  we denote the corresponding canonical basis of $\cK_3$ by
\be
\big\{x\otimes a\equiv |x\ket\otimes |a\ket ,\ \ x\in \cT_3, a\in A_3\big\}.
\ee

The dynamics of the quantum walker we consider, dubbed a coined QW, is given by the composition of a unitary update of the coin variables in $\C^3$ followed by a coin state dependent shift on the tree. 

\medskip

Let $C\in U(3)$ be a unitary matrix. The unitary update operator defined by $\I\otimes C$ which acts on the canonical basis of $\cK_3$ as
\be\label{reshuffle}
(\I\otimes C) x\otimes \tau=|x\ket\otimes |C\tau\ket=\sum_{\sigma \in A_3} C_{\sigma \tau}\, x\otimes \sigma,
\ee
where $\{C_{\sigma\tau}\}_{(\sigma,\tau)\in A_3^2}$ denotes the matrix 
$C$ in the ordered canonical basis of $\C^3$ above. 
In order to express the coin state -dependent shift $S$ on $\cK_3=\cT_3\otimes\C^3$,
three shifts on $l^2(\cT_3)$ are introduced. Let $x_e$, respectively $x_o$, denote vertices at even, respectively odd distance of the root. Such vertices will be called odd sites, respectiveley even sites in the sequel. For $a\neq b\in A_3$, we define $S_{ab}$ on $l^2(\cT_3)$ by
\be\label{sab}
{S_{ab}=\sum_{x_e \in \cT_q} |x_ea\ket\bra x_e |+ \sum_{x_o \in \cT_q} |x_o b\ket\bra x_o |. }
\ee
The operator $S_{ab}$ is unitary and such that $S_{ab}^*=S_{ab}^{-1}=S_{ba}$. The following immediate property justifies the name shift for $S_{ab}$. For each $x\in \cT_q$, consider $\cH^{ab}_x$ the $S_{ab}$-cyclic subspace generated by $|x\ket$, 
\be\label{hab}
\cH^{ab}_x=\mbox{span }\big\{S_{ab}^n |x\ket, \ n\in\Z\big\}\subset l^2(\cT_3).
\ee
\begin{lem}
The subspace $\cH^{ab}_x$ is isomorphic to $l^2(\Z)$ and $S_{ab}$ is unitarily equivalent to the shift on $l^2(\Z)$.
\end{lem}
With $S_{bc}$ and $S_{ca}$ defined similarly on $l^2(\cT_3)$,  the coin state dependent shift on $\cK_3$ is defined as  
\be\label{circa}
S=S_{bc}\otimes |a\ket\bra a|+S_{ca}\otimes |b\ket\bra b|+S_{ab}\otimes |c\ket\bra c|=\sum_{\circa}S_{ab}\otimes |c\ket\bra c|,
\ee
where $\circa$ indicates that all circular permutations the ordered set $\{a,b,c\}$ appear.
The unitary operator $U(C)$ describing the coined QW on $\cK_3$ is defined as
$
U(C)=S(\un\otimes C),
$
where the coin matrix $C\in U(3)$ is viewed as a parameter of the QW. More explicitly,  the action of $U(C)$ reads 
with (\ref{sab}),
\be\label{cx}
U(C)=\sum_{\circa}\left(\sum_{x_e \in \cT_q} |x_ea\ket\bra x_e |\otimes |c\ket\bra c| C+ \sum_{x_o \in \cT_q} |x_o b\ket\bra x_o |\otimes |c\ket\bra c| C\right) .
\ee
\begin{rem} For coined QWs defined on general homogeneous trees $\cT_q$, with $q\geq 3$ and more properties, see \cite{HJ2}
\end{rem}

Let us note here one symmetry property of the model. Other symmetries are expressed in Proposition \ref{symcyc}. 
For $z\in \cT_3$, let  $T_z$ be the isometric simply transitive map $\cT_3\ra \cT_3$ defined by $T_z x=zx$. Using the same notation for the corresponding operator acting on $l^2(\cT_3)$, we have that $T_z^{-1}=T_{z^{-1}}=T_z^*$ on $l^2(\cT_3)$. For $a, b \in A_3$
\bea\label{transl}
T_z^* S_{ab} T_z= S_{ab} \ \ \mbox{if $|z|$ is even}, \ \ \ \
T_z^* S_{ab} T_z= S_{ba} \ \ \mbox{if  $|z|$ is odd}.
\eea
Similar results hold for the other shifts. In particular, extending $T_z$ to $\cK_3$, we have
\be\label{comtz}
[S, T_z\otimes \un]=0, \ \ \mbox{if $|z|$ is even}.
\ee

\begin{rem}
The notation (\ref{cx}) allows us to consider that each site $x$ of the tree carries a coin matrix $C(x)\in U(3)$, which, in this case, is identical on each site $C(x)=C$ for all $x\in \cT_3$. For later purposes, following \cite{JM}, \cite{J4}, \cite{HJ}, we consider $\cC=\{C(x)\in U(3), x\in \cT_3\}$ a collection of coin matrices and consider
\be\label{ucc}
U(\cC)=\sum_{\circa}\left(\sum_{x_e \in \cT_q} |x_ea\ket\bra x_e |\otimes |c\ket\bra c| C(x_e)+ \sum_{x_o \in \cT_q} |x_o b\ket\bra x_o |\otimes |c\ket\bra c| C(x_0)\right) .
\ee
This is a well defined unitary operator on $\cK_3$, for any collection $\cC$.
\end{rem}

\subsection{Rooted  Binary Trees}

Making use of definition (\ref{ucc}), boundary conditions which preserve unitarity and restrain the configuration space of the motion of the walker can be defined, see \cite{HJ2}. In particular, the motion of the walker can be confined to the rooted binary tree $\cT_B$, with associated Hilbert space $\cK_B$ we now describe.

\medskip

Denote by ${\mathfrak S}_3$ the set of all permutations of the labels of $A_3$ and let 
 $\pi=(acb)\in {\mathfrak S}_3$ be the anti-cylic permutation. Consider the corresponding permutation matrix in the ordered basis $\{|a\ket, |b\ket, |c\ket\}$
 \be
 C_{\pi}=\begin{pmatrix}
 0 & 1 & 0\cr 
 0 & 0 & 1\cr
 1 & 0 & 0
 \end{pmatrix}.
 \ee
 Let $C\in U(3)$ be given and let $e\in \cT_3$ be the root.  We  define a site-dependent collection of matrices $\cC_{e}=\{C(x)\in U({3})\}_{x\in\cT_3}$ by
\be\label{cepi}
C(x)=\left\{\begin{matrix} 
C_{\pi} & \mbox{if} \ |x|\leq 1 \cr
C\ \ & \mbox{otherwise.\ \ \  \ \ }
\end{matrix}\right.
\ee
and consider $U(\cC_{e})$ defined by (\ref{ucc}).
As observed in \cite{HJ2}, the subspace
\bea\label{he}
\cH_{e}=
{\mbox{span }}\{e\otimes a, a\otimes c, e\otimes b, b\otimes a, e\otimes c, c\otimes b\},
\eea
is invariant under $U(\cC_{e})$ and 
$
\sigma(U(\cC_{e})|_{\cH_{e}})=\{1, e^{i\pi/3}, \cdots, e^{i5\pi /3}\}.
$ Moreover, the three infinite dimensional subspaces $\cH^{a}$, $\cH^{b}$ and $\cH^{c}$ given by
\be\label{ha}
\cH^{a}=\mbox{span }\{a \otimes a, a \otimes b\}\cup\{ay\otimes a, ay\otimes b, ay\otimes c\}_{ |ay|>|y|\geq 1}
\ee
and by circular permutation of the indices for $\cH^{b}$ and $\cH^{c}$,  are all invariant under $U(\cC_{e})$. This is due to the fact that the QW couples nearest neighbors on $\cT_3$ only, and that the subspaces $\cH^\#$, $\#\in\{a,b,c\}$, are separated by $\cH_e$ which is invariant.
Actually, each of the subspaces $\cH^\#$ is a direct sum of two infinite dimensional subspaces invariant under $U(\cC_{e})$, as easily checked.
\begin{lem} The following decomposition holds
\be\label{haa}
\cH^{a}=\cH_{a\otimes a}\oplus \cH_{a\otimes b},
\ee
where the subspaces $\cH_{a\otimes a}$ and $\cH_{a\otimes b}$  invariant under $U(\cC_{e})$ and given
\bea
\cH_{a\otimes a}&=&\mbox{span }\{a \otimes a\}\cup\{aby\otimes a, aby\otimes b, aby\otimes c\}_{ |aby|\geq |y|+2\geq 2}\nonumber\\
\cH_{a\otimes b}&=&\mbox{span }\{a \otimes b\}\cup\{acy\otimes a, acy\otimes b, acy\otimes c\}_{ |acy|\geq|y|+2\geq 2}.
\eea
Permutation of indices yield similar invariant decompositions for $\cH^{b}$ and $\cH^{c}$.
\end{lem}

Let us focus on the index $a$.  We denote by $U_a(C)$ the restriction $U(\cC_{e})|_{\cH^a}$ that we view as a QW on a binary tree $\cT_B^a$ with root $a$ going forward in the direction $a$, with coin space of dimension $3$ over each site of this rooted tree, except over the root $a$ where the coin space is of dimension $2$.
In other words
\be
 \cT^a_B=\{a\}\cup_{ |ay|> |y|\geq 1}\{ay\} 
 \ee
with corresponding Hilbert space $l^2(\cT^a_B)$ and  $\cK^a_B=\cH^a$ which depend on $a$, as subsets of $l^2(\cT_3)$ and $\cK_3$. See figure \ref{binarytree}.
\begin{figure}[htbp]
   \begin{center}
      \includegraphics[scale=.4, angle=-90]{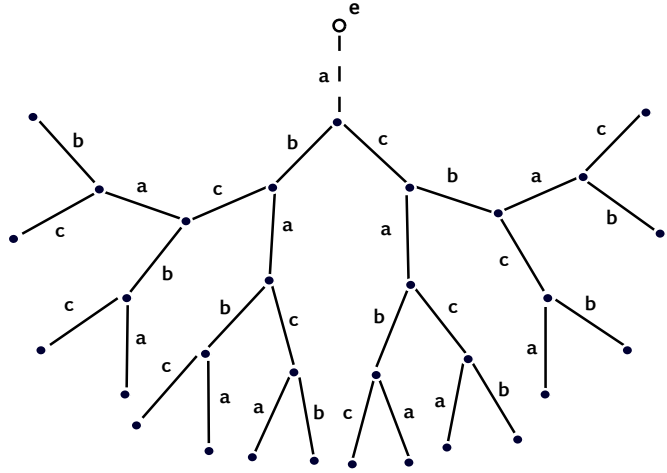}
   \end{center}
   \caption{\footnotesize The binary tree $\cT^a_B$, with sites indicated as black dots.}\label{binarytree}
\end{figure}
By definition, the coin matrix at each site of $\cT^a_B$ is given by $C$, except at the root which carries a two-dimensional coin state, where the boundary condition states that the update of the two basis coin states $|a\ket, |b\ket$ is carried out by means of $C_{\pi}$: 
\be
U_a(C) \ a\otimes a= ab\otimes c \ \ \ \mbox{and} \ \ \
U_a(C) \ a\otimes b= ac\otimes a, 
\ee
where
\be
\cK^a_B=\C_{a}^2\oplus_{y\in \cT^a_B\setminus \{a\} }\C_{y}^3.
\ee
Taking into account the finer decomposition (\ref{haa}), we define
\be
 \cT^{ab}_B=\{a\}\cup_{ |aby|\geq |y|+2 \geq 2}\{aby\}
 \ee
 the tree rooted at $a$ such that $a$ has coordination number one to $ab$, and all other sites have coordination number $3$,  see Figure \ref{binaries}.
 \begin{figure}[htbp]
   \begin{center}
      \includegraphics[scale=.4, angle=-90]{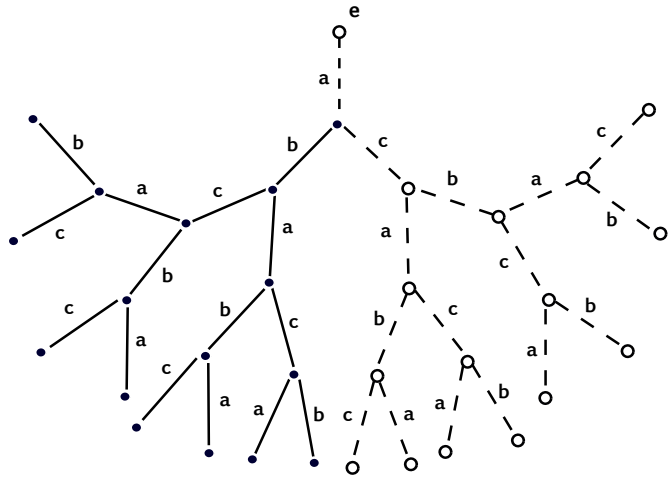}
   \end{center}
   \caption{\footnotesize The rooted tree $\cT^{ab}_B$, with sites indicated as black dots.}\label{binaries}
\end{figure}
 The corresponding  configuration and total
 Hilbert spaces are $l^2(\cT^{ab}_B)$ and $ \cK^{ab}_B=\cH_{a\otimes a}$ respectively, where
 \be
 \cK^{ab}_B=\C_a\oplus_{y\in \cT^{ab}_B\setminus \{a\}}\C_{y}^3,
 \ee 
 which are subsets of $l^2(\cT_3)$ and $\cK_3$. The tree $ \cT^{ac}_B$ and Hilbert spaces $l^2(\cT^{ac}_B)$ and  $\cK^{ac}_B=\cH_{a\otimes b}$ are defined similarly. We view $U_{ab}(C)=U(\cC_{e})|_{\cH_{a\otimes a}}$ and $U_{ac}(C)=U(\cC_{e})|_{\cH_{a\otimes b}}$ as QWs on $\cK_B^{ab}$ and $\cK_B^{ac}$, such that
 \be\label{deco}
 U_a(C)=U_{ab}(C)\oplus U_{ac}(C) \ \ \mbox{on} \ \ \cK_B^{a}=\cK_B^{ab}\oplus\cK_B^{ac}.
 \ee
 The boundary condition at the root $a$ of $\cT^{ab}_B$ and $\cT^{ac}_B$ then read 
 \be
 U_{ab}(C) a\otimes a= ab\otimes c, \ \ \ U_{ac}(C) a\otimes b= ac\otimes a.
 \ee
 This yields the 
\begin{definition}\label{bc1}
A coined QW on the  rooted tree $\cT^{ab}_B$ is defined by $U_{ab}(C)$ on $\cK_B^{ab}$, whereas a QW on the binary tree $\cT^a_B$ is defined by $U_a(C)$ on $\cK^a_B$.
 \end{definition}
\begin{rems}{\ \\}
i) The QW defined by $U_a(C)$ is a direct sum of independent QWs on $\cK_B^{ab}$ and $\cK_B^{ac}$, according to (\ref{deco}). We discuss  ways to couple them at the root in section \ref{fbc}.\\
ii) Similar interpretations hold for the restrictions $U_b(C)=U(\cC_{b})|_{\cH^b}$ and $U_c(C)=U(\cC_{c})|_{\cH^c}$.\\
iii) While we will not consider such generalizations, it is possible to decorate the entries of the matrix $C_\pi$ which define the boundary conditions by independent phases.  
\end{rems}

The QWs $U_\#(C)$ defined on $\cK_B^\#$ for $\#\in\{a,b,c\}$ are related to one another as the following proposition shows.
\begin{prop}\label{symcyc} Let $\sigma=(abc) \in \fS_3$ and $C_\sigma \in U(3)$ the corresponding permutation matrix. 
There exists a unitary operator $V$ on $\cK_3$ such that $V^3=\un$ and
\bea
V(\cK_B^\#)&=&\cK_B^{\sigma(\#)}, \ \ \ \mbox{for all $\#\in\{a,b,c\}$,}\nonumber \\
U_\#(C)&=&V^{-1}U_{\sigma(\#)}(C_\sigma CC^{-1}_\sigma)V.
\eea
In particular, we have on $\cK_3$,
\be
U( C )=V^{-1}U(C_\sigma CC^{-1}_\sigma)V.
\ee 
Also, with $V_a=T_a^{-1}VT_a$, one has $V_a(\cK_B^{ab})=\cK_B^{ac}$ and
\be
U_{ab}( C )= V_a^{-1}U_{ac}(C_\sigma CC^{-1}_\sigma)V_a.
\ee
\end{prop}
\begin{cor} For any $C\in U(3)$
\bea
&&\sigma(U( C ))= \sigma(U(C_\sigma CC^{-1}_\sigma)), \ \ \sigma(U_\#( C ))= \sigma(U_{\sigma(\#)}(C_\sigma CC^{-1}_\sigma)), \nonumber\\
&& \sigma(U_{ab}( C ))= \sigma(U_{ac}(C_\sigma CC^{-1}_\sigma)).
\eea
\end{cor}
\begin{rems}{\ \\}
i) All statements remains true if $\sigma$ is replaced by $\pi$, due to the relation $\sigma^2=\pi$.\\ 
ii) This symmetry shows that on the homogeneous tree, coin matrices that are unitarily equivalent by means of $C_\sigma$ or $C_\pi$  give rise to QWs with identical spectrum. On the binary tree, this property remains true provided the QW takes place on a different binary tree.  
\end{rems}
\begin{proof}
Let $\Sigma: \cT_3\mapsto \cT_3$ be defined by $\Sigma (x_1 x_2 \cdots x_n)=\sigma(x_1)\sigma(x_2)\cdots \sigma(x_n)$, for any reduced word $x_1 x_2 \cdots x_n\in \cT_3$, where $x_i\in A_3$. The inverse of $\Sigma$ is $\Pi$ s.t.  $\Pi (x_1 x_2 \cdots x_n)=\pi(x_1)\pi(x_2)\cdots \pi(x_n)$. We keep the same notation for the corresponding unitary operator on $l^2(\cT_3)$ defined by
\be
\Sigma |x_1 x_2 \cdots x_n\ket=|\sigma(x_1)\sigma(x_2)\cdots \sigma(x_n)\ket, 
\ee
for any basis vector $|x_1 x_2 \cdots x_n\ket$. The equivalent definition holds for $\Pi: l^2(\cT_3)\ra l^2(\cT_3)$. Then, by construction, $\Sigma \cT_B^\#=\cT_B^{\sigma(\#)}$, for all $\#\in\{a,b,c\}$. 
Consider $S_{ab}$. By looking at the action of the basis vectors $|x\ket\in l^2(\cT_3)$, one gets
$
S_{ab} \Sigma = \Sigma S_{\pi(a)\pi(b)} 
$
and similarly for circular permutations of indices. Define now the unitary operator on $\cK_3$
\be\label{defv}
V=\Sigma\otimes C_\sigma.
\ee
Restricting attention to vectors  in $\fB^+_2=\mbox{span }\{x\otimes \tau \ | \  x\in\cT_3, |x|\geq 2, \tau\in A_3\}$ where all coin matrices $U_\#(C)$ are equal to $C$, we have  
\bea
(S_{ab}\otimes |c\ket\bra c| C) V|_{\fB^+_2} &=&S_{ab}\Sigma \otimes  |c\ket\bra c| C C_\sigma|_{\fB^+_2}=\Sigma S_{\pi(a)\pi(b)} \otimes |\sigma(\pi(c))\ket\bra \sigma(\pi(c))| CC_\sigma|_{\fB^+_2}\nonumber\\
&=& (\Sigma \otimes C_\sigma ) S_{\pi(a)\pi(b)} \otimes |\pi(c)\ket\bra \pi(c)| C^{-1}_\sigma CC_\sigma|_{\fB^+_2}.
\eea
Summing over all permutations of the labels $\{a,b,c\}$, we obtain
\be\label{b2}
S (\un\otimes C)  (\Sigma\otimes C_\sigma)|_{\fB^+_2}=(\Sigma \otimes C_\sigma ) S (\un\otimes C^{-1}_\sigma C C_\sigma)|_{\fB^+_2}.
\ee
This argument actually shows that for $U(C)$ on $\cK_3$ , we have 
\be
U( C )=V^{-1}U(C_\sigma CC^{-1}_\sigma)V,
\ee and that 
\be
U_\#( C )|_{B^+_2\cap \cK_B^\#}=V^{-1}U_{\sigma(\#)}(C_\sigma CC^{-1}_\sigma)V|_{B^+_2\cap \cK_B^\#}.
\ee 
Consider now the action of $U_\#(C)$ on  the root of $\cK_B^\#$, {\em i.e.} on the vector  $\#\otimes \tau$, with $\#\in\{a,b,c\}$ and coin state $\tau\in\{\#,\sigma(\#)\}$. 
Let us compute $U_{\sigma(\#)}(C)V (\#\otimes \tau)$.  Since the coin matrix on the root is $C_\pi$, we have 
\bea
(S_{ab}\otimes |c\ket\bra c| C_\pi) V (\#\otimes \tau)&=&(S_{ab}\otimes |c\ket\bra c| C_\pi) (\sigma(\#)\otimes \sigma(\tau))=\sigma(\#)b \otimes |c\ket\bra c| C_\pi \sigma(\tau)\ket 
\nonumber\\
&=&   (\Sigma S_{\pi(a)\pi(b)} \#)\otimes (|\sigma(\pi(c))\ket\bra\sigma(\pi(c))|C_\pi\sigma(\tau)\ket\nonumber\\
&=& V (S_{\pi(a)\pi(b)}\otimes |\pi(c)\ket\bra \pi(c)| C^{-1}_\sigma C_\pi C_\sigma ) (\#\otimes \tau).
\eea
Summing over the permutations of $\{a,b,c\}$, and noting that $C^{-1}_\sigma C_\pi C_\sigma =C_\pi$ we have
\be
U_{\sigma(\#)}(C') V(\#\otimes \tau)=VU_\#(C'') (\#\otimes \tau).
\ee
Hence, for any $\#\in\{a,b,c\}$, and any $C\in U(3)$,
\be
U_\#(C)=V^{-1}U_{\sigma(\#)}(C_\sigma CC^{-1}_\sigma)V
\ee
on $\cK_B^{\#}$, with $V(\cK_B^{\#})=\cK_B^{\sigma\#}$. The proof of the statement about $U_{ab}( C )$ is quite similar.
\ep
\end{proof}

\subsection{Boundary Conditions}\label{fbc}

For illustration purposes, we introduce here a one-parameter families of boundary conditions at the root of the binary tree $\cT_B^a$ showing how to couple the two invariant subtrees $\cK_B^{ab}$ and $\cK_B^{ac}$. There are of course other possibilities.

\medskip

Consider 
\be
C_a^\theta=\begin{pmatrix}
\sin(\theta) & \cos(\theta) &  0 \cr
0 & 0& 1\cr
 \cos(\theta) & -\sin(\theta) & 0
\end{pmatrix}, 
C_b^\theta=\begin{pmatrix}
0& \cos(\theta) &  -\sin(\theta)  \cr
0 & \sin(\theta) &\cos(\theta) \cr
1 & 0& 0\cr
\end{pmatrix}, 
C_c^\theta=\begin{pmatrix}
0 & 1& 0\cr
 -\sin(\theta)&0& \cos(\theta)  \cr
\cos(\theta)  & 0 & \sin(\theta) 
\end{pmatrix}
\ee
and let $\cC^\theta={C^\theta(x)}_{x\in \cT_3}$ given by
\be
C^\theta(x)=\left\{\begin{matrix} 
C_\pi & \mbox{if} \ x=e \cr
C^\theta_{x} & \mbox{if} \ |x|= 1 \cr
C\ \ & \mbox{otherwise. \ \ \  \ \ }
\end{matrix}\right.
\ee
Direct computations establish the following.
\begin{lem}
The subspaces
$
\cH_{e}
$ (\ref{he})
and $\cH^{a}=\cK_B^a$ (\ref{ha})
are invariant under $U(\cC^\theta)$, for all $\theta\in \T$. Moreover, with 
\be
U^\theta_a(C)=U(\cC^\theta)|_{\cK_B^a},
\ee 
the boundary conditions at the root $a\in \cK_B^a$ read
\bea
U^\theta_a(C) \ a\otimes a&=& \cos(\theta) ab\otimes c + \sin(\theta)ac\otimes a \nonumber \\ 
U^\theta_a(C) \ a\otimes b&=& -\sin(\theta)ab\otimes c + \cos(\theta)ac\otimes a.
\eea
Finally,
\be
C_\sigma C^\theta_\# C^{-1}_{\sigma}=C^\theta_{\sigma(\#)}, \ \#\in\{a,b,c\},
\ee
so that for any $C\in U(3)$,  we have  $V(\cK_B^{\#})=\cK_B^{\sigma(\#)}$
\be
U^\theta_\#(C)=V^{-1}U^\theta_{\sigma(\#)}(C_\sigma CC^{-1}_\sigma)V \ \mbox{on $\cK_B^{\#}$.
}
\ee
 \end{lem}
\begin{rem}\label{uatheta}
By construction, $U^\theta_a(C)$ is a rank two perturbation of $U_a(C)$.
\end{rem}

Let us finally mention that $U^\theta_\#(C)$ 
is continuous in $C\in U(3)$ in the following sense. For all $C, C'\in U(3)$
\be
\|U^\theta_\#(C)-U^\theta_\#(C')\|\leq \|C-C'\|_{\C^3}.
\ee

In order to study the spectral properties of $U_a(C)=U_{ab}(C)\oplus U_{ac}(C)$ on the binary tree $\cK_B^{a}=\cK_B^{ab}\oplus\cK_B^{ac}$, we can restrict attention to the spectral measure of cyclic vectors generating $\cK_B^{ab}$ and $\cK_B^{ac}$. Such a vector exists when all matrix elements of $C$ are all different from zero. This  is true in particular for the circulant matrices we study below.
\begin{lem} If $C\in U(3)$ is such that $C_{\tau,\sigma}\neq 0$ for all $\sigma, \tau$,  the vector $a\otimes a$ at the root of $\cK_B^{a}$ is cyclic for $U_{ab}(C)$
and $a\otimes b$ is cyclic for $U_{ac}(C)$.
\end{lem}
\begin{proof} Consider $U_{ab}(C)$. It is enough to show that for all $x\otimes\tau\in \cK_B^{ab}$, there exists $j\in \Z$ s.t. $\bra x\otimes\tau |U_{ab}^j(C)a\otimes a\ket\neq 0$. Since all matrix elements of $U_{ab}(C)$ are non-zero, $|x|-1$ iterations of $U_{ab}(C)$ on $a\times a$  allow to reach the site $x\in \cT_B^{ab}$ along a specific edge which determines the coin state above this site. Two more iterations allow to reach $x$ again along the other two edges connected to $x$, which yield non zero components along the other two coin states above the same site $x$. 
\ep
\end{proof}

\section{Circulant Unitary and Orthogonal Matrices}

We shall restrict attention to the set of circulant coin matrices on $\C^3$, $\mbox{Circ}(3)$,  which allows for some simplifications in the analysis of the resolvent of $U_\#(C)$. Circulant coin matrices are such that the three QWs $U_\#(C)$ defined on $\cK_B^\#$ are unitarily equivalent and admit a convenient parametrization.
\medskip

We denote the set of $3\times 3$ unitary, respectively orthogonal, matrices by $U(3)$, respectively $O(3)$. 
Also, $\mbox{Circ}(3)$ denotes the set of $3\times 3$ circulant matrices in $M_3(\C)$,
\be
\mbox{Circ}(3)=\left\{\begin{pmatrix}
c_0 & c_2 & c_1 \cr c_1 & c_0 & c_2 \cr c_2 & c_1 & c_0
\end{pmatrix}=\mbox{circ}(c_0,c_1,c_2), \ c_j\in\C, j=0,1,2.\right\}.
\ee

Let us recall some properties of circulant unitary or orthogonal matrices to be used later.

\begin{lem} We have
\be
\{C\in M_3(\C) \ | \ C=C_\sigma C C^{-1}_{\sigma}\}=\mbox{\em Circ}(3)
\ee
For all $C\in \mbox{\em Circ}(3)\cap U(3)$ with $\sigma(C)=\{e^{i\theta_j}\}_{j=0,1,2}$, it holds with $\eps=e^{i2\pi/3}$,
\be \label{ccircu}
C=\frac13\mbox{\em circ}(e^{i\theta_0}+e^{i\theta_1}+e^{i\theta_2}, e^{i\theta_0}+\eps^2e^{i\theta_1}+\eps e^{i\theta_2}, e^{i\theta_0}+\eps e^{i\theta_1}+\eps^2e^{i\theta_2}).
\ee
Moreover, if $c_j=0$ for some $j=0,1,2$, then $C\in e^{i\alpha}\{\un, C_\sigma, C_\pi\}$, for some $\alpha\in \R$.
\bea\label{param12}
&&\mbox{\em Circ}(3)\cap O(3):=\mbox{\em CO}_+(3)\cup \mbox{\em CO}_-(3)=\nonumber\\
&&\left\{\mbox{\em circ}(c_0,c_1,c_2), \ \mbox{s.t.}\ \begin{pmatrix} c_0(t) \cr c_1(t) \cr c_2(t)  \end{pmatrix} =
\frac13\begin{pmatrix} 1+\sin(t)+\sqrt3\cos(t) \cr 1+\sin(t)-\sqrt3\cos(t) \cr 1-2\sin(t)   \end{pmatrix} t\in [0,2\pi)\right\}\cup
\nonumber\\
&&\left\{\mbox{\em circ}(c_0,c_1,c_2),  \ \mbox{s.t.}\ \begin{pmatrix} c_0(t) \cr c_1(t) \cr c_2(t)  \end{pmatrix} =
\frac13\begin{pmatrix} -1+\sin(t)+\sqrt3\cos(t) \cr -1+\sin(t)-\sqrt3\cos(t) \cr -1-2\sin(t)   \end{pmatrix} t\in [0,2\pi)\right\}.
\eea
\end{lem}
\begin{rem}\label{copcom}  The two disjoint pieces of $\mbox{\em Circ}(3)\cap O(3)$ are related by the identities $c_j(t+\pi)=-c_j(t)$, $j=1,2,3$, so that one case can be deduced from the other.
\end{rem}
\begin{proof} The first statement is a computation. The second follows from the well known fact that all circulant matrices can be diagonalized by the same unitary change of basis. Explicitly here, with $W=\frac{1}{\sqrt3}\begin{pmatrix} 1&1&1 \cr 1&\eps&\eps^2\cr 1&\eps^2&\eps\end{pmatrix}$ and $\eps=e^{i2\pi/3}$,
\be
W^{-1}\mbox{ circ}(c_0,c_1,c_2)W=\mbox{diag}(c_0+c_1+c_2, c_0 + \eps c_1 + \eps ^2 c_2, c_0 + \eps^2 c_1 + \eps  c_2).
\ee
Hence these matrices are parameterized by their eigenvalues. Imposing three eigenvalues on the unit circle yields the result. The following property is straightforward whereas the last statement can be obtained by expressing the orthogonality condition on $\mbox{ circ}(c_0,c_1,c_2)$ into geometric conditions on the real vector $\vec c=(c_0,c_1,c_2)^T$: $\|\vec c\|=1$ and $\vec c \cdot R_{\vec n}(4\pi/3) \vec c=0$, where $\vec n=\frac{1}{\sqrt3}(1,1,1)^T$, and $R_{\vec n}(\theta)$ is the rotation of angle $\theta$ of axis $\vec n$. Hence   $\vec c$ belongs to the intersection of the unit sphere with two planes orthogonal to $\vec n$, and passing through the points $\pm\frac{1}{3}(1,1,1)$.\ep
\end{proof}

\subsection{Special Cases}
 
If $C\in \{C_\omega, | \omega \in \fS_3\}\cap \mbox{Circ}(3)=e^{i\delta}\{\un, C_\sigma, C_\pi\}$, $\delta\in \R$, we have full understanding of the spectrum of $U^\theta_a(C)$ 
on the binary tree $\cK_B^a$:  
\begin{prop} For all $\theta \in \T$,
\bea
\sigma(U^\theta_a(e^{i\delta}\un))&=&\sigma_{ac}(U^\theta_a(e^{i\delta}\un))=\Ss \nonumber\\
\sigma(U^\theta_a(e^{i\delta}C_\pi))&=&e^{i\delta}\{e^{ik2\pi/6}\}_{k=0,1,\dots, 5}\cup\sigma_d(U^\theta_a(e^{i\delta}C_\pi)) \nonumber\\
\sigma(U^\theta_a(e^{i\delta}C_\sigma))&=&e^{i\delta}\{e^{ik2\pi/6}\}_{k=0,1,\dots, 5}\cup\sigma_d(U^\theta_a(C_\sigma))\nonumber\\
e^{i\delta}\{e^{ik2\pi/6}\}_{k=0,1,\dots, 5}&=&\sigma_{ess}(U^\theta_a(e^{i\delta}C_\pi))=\sigma_{ess}(U^\theta_a(e^{i\delta}C_\sigma))\nonumber
\eea
If $\theta =\delta=0$, 
\be
\sigma_d(U^0_a(\pm C_\#))=\emptyset, \ \ \#\in \{\sigma, \pi\}.
\ee
\end{prop}

\begin{rems}{\ \\} 
0) This proves the last statements of Theorem \ref{mainthm} below.\\
i)
If $C=C_\sigma$, on top of the six-dimensional invariant subspaces by 
$\cH_{T_za}$, $|z|$ even and s.t. $za\in \cT_B^a$, the following subspaces are invariant under $U_a(C_\sigma)$: $\mbox{span} \{a\otimes a, ab\otimes c\}$ and $\mbox{span} \{a\otimes b, ac\otimes a\}$.\\
ii) The discrete spectrum $\sigma_d(U^\theta_a(e^{i\delta}C_\#))$, $\#\in\{\pi, \sigma\}$ consist in twelve distinct eigenvalues at most and depends on $\theta$ in general.
\end{rems}
\begin{proof} \\ We start with $\delta=0$.
In case $C=\un$, $U_a(\un)$ acts as $S$ on sites  $x\otimes \tau \in \cK_B^a$ with $|x|\geq 2$. Therefore, all cyclic subspaces $\cH_y^{\#\sigma(\#)}$ of the form (\ref{hab}), with $\#\in\{a,b,c\}$, $y\in\cT_B^a$ are invariant under $U_a(\un)$, provided $a\otimes \tau\not\in \cH_y^{\#\sigma(\#)}$.  The restrictions of $U_a(\un)$ to these subspaces are all unitarily equivalent to a shift, so that $\sigma(U_a(\un))=\Ss$. Let $\cH_{a\otimes \tau}$ be the cyclic subspace generated by $a\otimes \tau$, $\tau\in\{a,b\}$. One has
\bea
\cH_{a\otimes a}&=&\overline{\mbox{span}}\{\cdots abcb\otimes a, abc\otimes a, ab\otimes a, a\otimes a, ab\otimes c, aba\otimes c, abab \otimes c, \cdots\}\nonumber \\
\cH_{a\otimes b}&=&\overline{\mbox{span}}\{\cdots acac\otimes b, aca\otimes b, ac\otimes b, a\otimes b, ac\otimes a, acb\otimes a, acbc \otimes a, \cdots\},
\eea
where the vectors are listed according to their image by $U_a(\un)$. Hence, the corresponding restrictions also give rise to shifts. Therefore, all restrictions to cyclic subspaces are absolutely continuous as well, which yields the result. \\
Consider now $U^{\theta}_a(\un)$. This operator differs from $U_a(\un)$ by a rank two perturbation which couples the two shifts induced by $U_a(\un)$ in $\cH_{a\otimes a}$ and $\cH_{a\otimes b}$. Each of these shifts is unitarily equivalent to a multiplication by $e^{ix}$ on $L^2(\T)$, in Fourier space, where $e^{ix}$ admits an analytic continuation  in a complex neighborhood of $\T$.  We can thus apply the argument of the proof of Theorem 6.2 in \cite{BHJ} to deduce that for any $\theta\in\T$, $\sigma_{sc}(U^{\theta}_a(\un))=\emptyset$.\\
Finally, a direct argument proves the absence of eigenvalues. Namely, mapping $\cH_{a\otimes a}$, respectively $\cH_{a\otimes b}$, to $\overline{\mbox{span}}\{|2j\ket, j\in \Z\}$, respectively $\overline{\mbox{span}}\{|2j+1\ket, j\in \Z\}$, with $|0\ket=a\otimes a$ and $|1\ket=a\otimes b$, the restriction $U^{\theta}=U^{\theta}_a(\un)|_{\cH_{a\otimes a}\oplus\cH_{a\otimes b}}$  reads
\bea
&&U^{\theta}|2j\ket=|2j+2\ket, \ \ U^{\theta}|2j+1\ket=|2j+3\ket, \forall j\neq 0\nonumber\\
&&U^{\theta}|0\ket= \cos(\theta)|2\ket +\sin(\theta)|3\ket, \ \ U^{\theta}|1\ket= -\sin(\theta)|2\ket+\cos(\theta)|3\ket. 
\eea
The eigenvalue equation $U^\theta \psi=\lambda\psi$ with $|\lambda|=1$ has no non trivial  $l^2$ solution, whereas the restrictions of $U^{\theta}_a(\un)$ to the cyclic subspaces $\cH_y^{\#\sigma(\#)}$ yields absolutely continuous shifts.\\
We address now $C\in\{C_\pi, C_\sigma\}$.  The essential spectrum of $U_a^\theta(C_\pi), U_a^\theta(C_\sigma)$ is dealt with as in \cite{HJ2}. The discrete spectrum comes from the restrictions to the 12 dimensional invariant subspaces generated by vectors at the root. We consider $U_a^\theta(C_\pi)$ only, the other case being similar, and simply check that
\bea\label{haay}
&&\cH_{a\otimes a}=
\mbox{span}\{a\otimes a, ab\otimes c, abc\otimes b, ab\otimes a, aba\otimes c, ab\otimes b, \\ \nonumber
&& \phantom{yyyyyyyyyyyyyyyyyyyyyyy}
ac\otimes a, aca\otimes c, ac\otimes b, acb\otimes a, ac\otimes c, a\otimes b\}\nonumber
\eea
is invariant. 
When $\delta\neq 0$, we note that 
\bea\label{phaserankone}
e^{-i\delta}U_{ab}(e^{i\delta}C)&=&U_{ab}(C)+(e^{-i\delta}-1)U_{ab}(C)|a\otimes a\ket\bra a\otimes a |\nonumber\\
e^{-i\delta}U_{ac}(e^{i\delta}C)&=&U_{ac}(C)+(e^{-i\delta}-1)U_{ac}(C)|a\otimes b\ket\bra a\otimes b| 
\eea
which together with (\ref{deco}), Remark \ref{uatheta} shows that 
$e^{-i\delta}U_a(e^{i\delta}C)$ is a rank two perturbation of both $U^\theta_a(C)$ 
for any value of $\theta$ and $\alpha$, such that the range of the perturbation is spanned by $\{ab\otimes c, ac\otimes a\}$. This is enough to get the result.
\ep
\end{proof}

\section{Spectral Analysis}

\subsection{Spectral Measure}

The spectral properties of a unitary operator $U$ on a Hilbert space $\cH$ admitting a normalized cyclic vector $\ffi$ can be read off the spectral measure of this vector,  $d\mu(\cdot)$, on the circle $\T$. We recall the properties of such probability measures we shall use below. For proofs, see {\em e.g.} \cite{Si}.

\medskip

Consider the decomposition of the measure into its absolutely continous and singular part with respect to the Lebesgue measure on $\T$ 
\be 
d\mu(\theta)=\frac{w(\theta)}{2\pi}d\theta+d\mu_s(\theta).
\ee
The Carath\'eodory function of $d\mu$ is defined for all $\D=\{z\ |\ |z|<1\}\subset \C$ by
\be\label{car}
F(z)=\int_\T\frac{e^{i\theta}+z}{e^{i\theta}-z}d\mu(\theta)
\ee
and satisfies $F(0)=1, \mbox{Re} F(z)>0$. The boundary values of $F$ allow us to recover the measure according to 
\bea
&&\lim_{r\ra 1^-}F(re^{i\theta})=F(\theta) \ \ \mbox{exists $\frac{d\theta}{2\pi}$ a.e.}\\
&&w(\theta)=\mbox{Re} F(\theta)\\
&&d\mu_s  \  \mbox{is supported on }  \{ e^{i\theta} | \lim_{r\ra 1^-}\mbox{Re} F(re^{i\theta})=\infty\}\\ \label{atom}
&&\mu(\{\theta_0\})=\lim_{r\ra 1^-}\frac{1-r}{2}F(re^{i\theta_0}).
\eea

The Carath\'eodory function is related to the resolvent of $U$, $G(z)=(U-z)^{-1}$, and to $H(z)=U(U-z)^{-1}$ for $z\in \D$ by
\be\label{rescara}
F(z)=1+2z\bra \ffi |G(z) \ffi\ket=2\bra \ffi |H(z) \ffi\ket-1.
\ee
\begin{rem}\label{bddres}
If  $\lim_{r\ra1^-}\bra\ffi | H(re^{i\theta_0})\ffi\ket$ is bounded for $\theta_0\in \T$, then $\theta_0\not\in \mbox{supp } d\mu_s$.
\end{rem}

The main technical result of the paper leading to Theorem \ref{mainthm} below, reads as follows:
\begin{thm}\label{maintech}
Let $C\in \mbox{Circ}(3)\cap U(3)$ and write the Carath\'eodory function of the spectral measure $d\mu_{a\otimes a}$ as $F(z)=2g(z)-1$, for $z\in \D$. Then, there exists a
polynomial in $(g,x)$ of the form 
\be\label{implicit}
\Phi(g,x)=c_5(x)g^5+c_4(x)g^4+\cdots+c_0(x)
\ee
with $c_j$ of degree $\leq 3$ for $j=1, \dots, 5$, and $c_0$ of degree $\leq 2$, 
and 
\be\label{mabs}
M(g) = m_2g^2+ m_1g +1
\ee
with constant coefficients $m_2, m_1$
such that $g(z)$ satisfies 
\be\label{polphi}
\Phi(g(z),z^2)\equiv 0, \ \ \mbox{for all $z\in \D$ s.t.} \ \ M(g(z))\neq 0.
\ee
With the parametrization 
\be\label{para+1}
C = \mbox{\em circ}(\alpha, \gamma, \beta+1),
\ee
$M$ and $c_5$ read
\bea
M(g)&=&(\beta^2-\alpha\gamma)g^2+2\beta g+1,\\
\label{c5}
c_5(x)&=& -x^3\{(\beta^2-\alpha \gamma)^2 \nonumber\\
         &&+x^2\gamma(\beta^2-\alpha \gamma)(2 (\alpha^3+\beta^3+\gamma^3-3\alpha \beta\gamma)+3(\beta^2-\alpha\gamma)) \nonumber\\
&&+x(\alpha\beta+\alpha-\gamma^2)(\alpha^3+\beta^3+\gamma^3-3\alpha \beta\gamma+\beta^2-\alpha\gamma)\times\nonumber\\
&&\hspace{5cm}\times
(\alpha^3+\beta^3+\gamma^3-3\alpha \beta\gamma+3(\beta^2-\alpha\gamma)) \nonumber\\
&&+(\alpha+\beta+\gamma+1)(\alpha^2+\beta^2+\gamma^2-\beta\gamma-\alpha\gamma-\alpha\beta-\alpha-\gamma+2\beta+1)\times\nonumber\\
&&\hspace{5cm}\times(\alpha^3+\beta^3+\gamma^3-3\alpha \beta\gamma+\beta^2-\alpha\gamma)^2 \}.
\eea
\end{thm}
\begin{rems}{\ \\}\label{rembc} 
i) This is actually a result on $U_{ab}(C)$ on $\cK_B^{ab}$ generated by $a\otimes a$. Proposition \ref{symcyc} shows that it is enough to consider this case to study $U_a(C)$. We stick to the simpler notation $U_a(C)$. \\
ii) The Carath\'eodory function $F(z)$  depends on $z^2$, which implies a symmetry of the spectrum stated in Theorem \ref{mainthm}.
\\
iii)  The set of points in $\D$ such that $M(g(z))=0$ can only accumulate on $\mathbb S$.
\\
iv) The artificial $+1$ in the parametrization (\ref{para+1}) is introduced so that the matrix $C-C_\pi$ which appears in the proof takes a simpler form.\\
v) The explicit form of the other polynomials $c_j$'s is provided in Propositions \ref{opc} and \ref{cjcop}.\\
vi) The equation satisfied by the Carath\'eodory function comes from an implicit equation for the restriction of the resolvent to a finite dimensional subspace expressed as Corollary \ref{matimp}.
\end{rems}
This result yields a criterion for absolutely continuous spectrum.
\begin{cor}\label{critacucirc} For any $C\in \mbox{Circ}(3)\cap U(3)$, 
if the leading term polynomial coefficient in (\ref{implicit}) has no root on the unit circle, then $U_a(C)$ has purely absolutely continuous spectrum.
\end{cor}
\begin{proof} By continuity of the roots of polynomials in the coefficients, all roots of $\Phi(g,z^2)$, in $g$, are bounded for $z$ in a neighborhood of $\Ss$. Hence, for any $e^{i\theta}$, $\Re g(re^{i\theta})$ cannot tend to $+\infty$ as $r\ra 1^-$, even if $M(re^{i\theta})=0$ for infinitely many $r$'s. Therefore, $\mbox{supp }d\mu_s=\emptyset$.
\ep
\end{proof}
\medskip

The strategy of the proof of Theorem \ref{maintech}, which can be found in Appendix, consists in making use of the properties of the binary tree and of the symmetries provided by the choice $C\in \mbox{Circ}(3)\cap U(3)$, to obtain an implicit equation for the matrix element $\bra a\otimes a|(U_a(C)-z)^{-1} a\otimes a\ket$ of the resolvent. Due to the complexity introduced by the coin space ensuring unitarity of the QW, the computations are more involved than in the self adjoint case, where the same strategy yields a matrix element of the resolvent of the Laplacian in a few lines. 
\medskip

The spectral consequences of Theorem \ref{maintech} of the QWs we consider are the following:
\begin{thm}\label{mainthm}
For any $C\in \mbox{\em Circ(3)}\cap U(3)$, 
\be\label{mdr}
 \sigma_{sc}(U_a(C))=\emptyset, \ \ \mbox{and}\ \ \sigma(U_a(C))=-\sigma(U_a(C)).
\ee
If $C\in\mbox{\em Circ(3)}\cap O(3)\setminus\{\pm C_\sigma, \pm C_\pi\}$,
\be
\sigma(U_a(C))=\sigma_{ac}(U_a(C)), \ \ \mbox{and}\ \ \sigma(U_a(C))=-\sigma(U_a(C))=\overline{\sigma(U_a(C))},
\ee
whereas \vspace{-.5cm}
\bea
&&\sigma(U_a(\pm \un))=\sigma_{ac}(U_a(\pm \un))=\Ss, \nonumber\\
&&\sigma(U_a(\pm C_\#))=\sigma_{ess}(U_a(\pm C_\#))=\{e^{ik2\pi/6}\}_{k=0,1,\dots, 5}, \ \ \#\in\{\sigma, \pi\}.
\eea
\end{thm}
\begin{rems}{\ \\}
i) For matrices $C\in U(3)$ such that $\|C-\un\|_{\C^3}\leq \eps$, with $\eps>0$, $\sigma(U(C))=\sigma_{ac}(U(C))$, as shown in \cite{HJ2}. The argument carries over to QWs $U_a(C)$. \\
ii) The proofs of the statements about $C\in U(3)\cap \mbox{Circ}(3)$ are given in the present section, whereas those concerning $C\in O(3)\cap \mbox{Circ}(3)$ are given in Appendix A.
\end{rems}

\noindent{\bf Proof of Theorem \ref{mainthm}:}
Consider $C\in \mbox{Circ}(3)\cap U(3)\setminus \{\pm \un, \pm C_\sigma, \pm C_\pi\}$. By Theorem \ref{maintech}, the Carath\'eodory function is determined for $z\in \D$ by one of the roots of the polynomial $\Phi(g,z^2)$ of degree 5  in $g$, with polynomial coefficients $c_j(z^2)$, $j=0,\cdots, 5$. This implies that $g(z)=g(-z)$ so that $\Re g(re^{i\theta})=\Re g(re^{i(\theta+\pi)})$. In the limit $r\ra 1^-$, we get $d\mu(\theta)=d\mu(\theta+\pi)$, which proves the symmetry of the spectrum expressed in (\ref{mdr}). Since the ${c_j}$'s are polynomials in $z$, the roots of $\Phi(g,z^2)$ are continuous in $z$ and remain bounded as long as $z$ belongs to $\C\setminus Z_5$, $Z_5\subset \C$ being the set of roots of $c_5(z^2)$.
By the argument used in the proof of Corollary \ref{critacucirc}, $d\mu_s$ is supported on the finite set $Z_5\cap \Ss$, so that it can consist of atoms only. This shows that $\sigma_{sc}(U_a(C))=\emptyset$ and the same result holds for $U_b(C)$ and $U_c(C)$.  The proof of the statements regarding $C\in O(3)\cap \mbox{Circ}(3)$ requires more detailed informations about $\Phi(g,z^2)$ and is given as Proposition \ref{eop} in Appendix A. 
\ep
\begin{rem} By construction, 
$U(C)=U(\cC_e)+F$, where $F$ is finite rank, hence trace class, and $U(\cC_e)$ is the direct sum of the $U_\#(C)$ and $U(C_\pi)|_{\cH_e}$. 
Therefore, by Birman-Krein theorem, for any $C$,
\be \sigma_{ac}(U(C))=\sigma_{ac}(U(\cC_e)) \ \ \mbox{and} \ \ U(C)|_{ac}\simeq U(\cC_e)|_{ac},
\ee 
where $U|_{ac}$ is the restriction of the unitary operator $U$ to its absolutely continuous subspace and $\simeq$   denotes unitary equivalence.
\end{rem}

\section{Appendix}
\subsection{Proof of Theorem \ref{maintech}}

Our goal is to find a closed equation for the restriction of the resolvent $U_a(C)$ to a finite dimensional subspace, making use of the symmetries of the operator and of the self similar structure of the tree. The strategy is to relate $U_a(C)$ and $U_a(\cC^\pi)$ defined by $U(\cC^\pi)|_{\cK_B^a}$ where $\cC^\pi$ is given by
\be\label{cepipi}
C(x)=\left\{\begin{matrix} 
C_{\pi} & \mbox{if} \ |x|\leq 3 \cr
C\ \ & \mbox{otherwise,\ \ \  \ \ }
\end{matrix}\right.
\ee
to be compared with (\ref{cepi}) defining  $U_a(C)$.
\medskip

For any element $x\otimes {\tau}$ of $\mathcal{K}_B^a$, we denote the projection on the site  $x \otimes {\tau}$ by $P_x^{\tau}$. We also write $P_x^{a,b} = P_x^a + P_x^b$ (and $\circa$) and $P_x =  P_x^a + P_x^b + P_x^c$.\\
We set:
\be \nonumber
P_0^+ = P_a^a 
,\;\quad\;
P_1^- = P_{ab}^c 
,\; \quad \;
P_1^+ = P_{ab}^{b,a}
,\;\quad\;  
P_1  = P_1^- +P_1^+,
\ee
\be\nonumber
P^-_2 =  P_{aba}^c+ P_{abc}^b ,
\quad 
P_2^+ =  P_{aba}^{a , b}  + P_{abc}^{c, a} ,
\quad 
P_3^- = P_{abab}^c +P_{abac}^a + P_{abca}^b +P_{abcb}^c 
\ee
\be\label{defproj}
P_l = P_0^+ + P_1 +P_2^-, \quad
P_d = P_2^+ + P_3^-.
\ee
The index gives the distance of the site to the root $a \otimes a$ and $P_l$ is the projection on a 6-dimensional invariant subspace of $U_a(\cC^\pi)$, which is orthogonal to $P_d$.

\begin{lem}\label{l1}
 $U_a(C)-U_a(\cC^\pi)=S(\un \otimes(C-C_\pi))(P_1+P_2)$ satisfies
\be
S(\un \otimes(C-C_\pi))P_k=(P_{k-1}^+ + P_{k+1}^- )S(\un \otimes(C-C_\pi))P_k, \ k=1,2.
\ee
\end{lem}

\proof
We compute
$
S(\un \otimes(C-C_\pi)) ab \otimes \tau = \lambda_1 a \otimes a + \lambda_2 abc \otimes b + \lambda_3 aba \otimes c
$,
where the $\lambda_j$ are constant depending on $\tau$ and $C$. Each vector on the right hand side belongs to $P_0^+$ or $P_2^-$. The same computation on $ac \otimes \tau$ yields the same conclusion and the case $k=2$ is dealt with similarly. \ep

We use the shorthands $U=U_a(C)$, $U_\pi=U_a(\cC^\pi)$, $G^z=(U-z)^{-1}$, $G_\pi^z=(U_\pi-z)^{-1}$ and $S_\pi=S(\un \otimes(C-C_\pi))(P_1+P_2)$ and consider the resolvent equation for $z\in \D$,
\be\label{reseq}
G^z-G_\pi^z = - G^z S_\pi G_\pi^z .
\ee
Making use of Lemma \ref{l1} and of the invariance of $P_l  \mathcal{K}_B^a$ under $G^z_\pi$, we get 
\begin{equation}\label{equa-res1}
P_l G^z P_l = P_l G^z_\pi P_l - P_l G^z P_l S_\pi P_l G^z_\pi P_l - P_l G^z  P_3^- S_\pi P_2^- G^z_\pi P_l .
\end{equation} 
We want to compute $P_l G^z P_l $, so the term $P_l G^z P_3^- $ needs to be transformed.  Assuming for now that the restriction $(\un+ P_3^- S_\pi P^+_2 G^z_\pi P_3^-)^{-1}|_{P_3^-  \mathcal{K}_B^a}$ exists, (\ref{reseq})  further yields
\be
P_l G^z P_3^- = -  P_l G^z P_1^+S_\pi P^+_2G^z_\pi P_3^- (\un + P_3^- S_\pi P^+_2 G^z_\pi P_3^-)^{-1}.
\ee
We eventually get
\be\label{matriximplict}
P_lG^zP_l =  P_lG^zP_l \times A(G^z_\pi) + P_lG^z_\pi P_l 
\ee
where 
\bea\label{defag}
A(G^z_\pi) &=& P_1^+S_\pi P^+_2G^z_\pi P_3^- (\un + P_3^- S_\pi P^+_2 G^z_\pi P_3^-)^{-1} P_3^- S_\pi P_2^- G^z_\pi P_l
\nonumber\\
&&
-
P_0^+ S_\pi P_1 G^z_\pi P_l 
-
P_1^+ S_\pi P_2^- G^z_\pi P_l
-
 P_2^- S_\pi P_1 G^z_\pi P_l.
\eea
The term $P_lG^z_\pi P_l$ which appears in $A(G^z_\pi)$ can be made explicit:
\begin{lem} In the basis   $\{a\otimes a, ab\otimes c, abc\otimes b, ab\otimes a, aba\otimes c, ab\otimes b, 
 \}$  
we have
$$
P_lG_\pi^z P_l = P_l(U_\pi - z)^{-1}P_l =  \frac{1}{1-z^6}
\begin{pmatrix}
z^5 & 1   & z   & z^2 & z^3 & z^4     \\
z^4 & z^5 & 1   & z   & z^2 & z^3   \\
z^3 & z^4 & z^5 & 1   & z   & z^2   \\
z^2 & z^3 & z^4 & z^5 & 1   & z      \\
z   & z^2 & z^3 & z^4 & z^5 & 1      \\
1   & z   & z^2 & z^3 & z^4 & z^5   
\end{pmatrix} 
$$
\end{lem}
In order to obtain an implicit equation for $P_l G^z P_l$, we show that the factor $P^+_2G^z_\pi P_3^-$ which cannot be computed explicitly, can be expressed in terms of the restricted resolvent $P_l G^z P_l$:
\begin{prop}\label{a=ag} We have 
\bea
&&P_2^+ G^z_\pi P_3^- = \langle a \otimes a | G^z   ab \otimes c \rangle \Big(|aba \otimes a\ket\bra abab \otimes c |+|aba \otimes b\ket\bra abac \otimes a| \\
\nonumber 
&& \hspace{6cm}+|abc \otimes c\ket\bra abca \otimes b|+|abc \otimes a\ket\bra abcb \otimes c|\Big)
\eea

\end{prop}

{\bf Proof:} We view $P^+_2G^z_\pi P_3^-$ as a $4 \times 4$  matrix from $P_3^-\cK_B^a$ to $P_2^+\cK_B^a$ in the bases defined according to the order given in the definition of the projectors (\ref{defproj}). By definition, 
\bea
&&P_2^+ G^z_\pi P_3^- = \\ \nonumber
&&
\hspace{-1cm} \begin{pmatrix}
\langle aba \otimes a |G^z_\pi abab \otimes c \rangle & 
\langle aba \otimes a |G^z_\pi abac \otimes a \rangle&
\langle aba \otimes a |G^z_\pi abca \otimes b \rangle&
\langle aba \otimes a |G^z_\pi abcb \otimes c \rangle \\
\langle aba \otimes b |G^z_\pi abab \otimes c \rangle & 
\langle aba \otimes b |G^z_\pi abac \otimes a \rangle&
\langle aba \otimes b |G^z_\pi abca \otimes b \rangle&
\langle aba \otimes b |G^z_\pi abcb \otimes c \rangle \\
\langle abc \otimes c |G^z_\pi abab \otimes c \rangle & 
\langle abc \otimes c |G^z_\pi abac \otimes a \rangle&
\langle abc \otimes c |G^z_\pi abca \otimes b \rangle&
\langle abc \otimes c |G^z_\pi abcb \otimes c \rangle \\
\langle abc \otimes a |G^z_\pi abab \otimes c \rangle & 
\langle abc \otimes a |G^z_\pi abac \otimes a \rangle&
\langle abc \otimes a |G^z_\pi abca \otimes b \rangle&
\langle abc \otimes a |G^z_\pi abcb \otimes c \rangle 
\end{pmatrix}
\eea
The presence of $C_\pi$ on the site $ab$ decouples paths starting from this point on and thus entries with both $aba$ and $abc$ are equal to zero.
Thus, 
\bea
&&P_2^+ G^z_\pi P_3^- = \\ \nonumber
&&\hspace{-1cm}  \begin{pmatrix}
\langle aba \otimes a |G^z_\pi abab \otimes c \rangle & 
\langle aba \otimes a |G^z_\pi abac \otimes a \rangle&
0&
0 \\
\langle aba \otimes b |G^z_\pi abab \otimes c \rangle & 
\langle aba \otimes b |G^z_\pi abac \otimes a \rangle&
0&
0 \\
0 & 
0&
\langle abc \otimes c |G^z_\pi abca \otimes b \rangle&
\langle abc \otimes c |G^z_\pi abcb \otimes c \rangle \\
0 & 
0&
\langle abc \otimes a |G^z_\pi abca \otimes b \rangle&
\langle abc \otimes a |G^z_\pi abcb \otimes c \rangle 
\end{pmatrix}
\eea 
Now, using the general property for $|z|$ even 
$
T_z^{-1}U(\cC)T_z=U(\cC_z)
$
where the configuration $\cC_z$ is defined from $\cC=\{C(x)\}_{x\in \cT_3}$
by $\cC_z=\{C(zx)\}_{x\in \cT_3}$, which follows from (\ref{comtz}), we can identify
the upper left block with
\bea
P_0^+ (U_a - z)^{-1} P_1^- &=&
\begin{pmatrix}
\langle a\otimes a | (U_a - z)^{-1}  ab \otimes c \rangle & \langle a\otimes a | (U_a - z)^{-1}  ac \otimes a \rangle\\
\langle a\otimes b | (U_a - z)^{-1}  ab \otimes c \rangle& \langle a\otimes b | (U_a - z)^{-1}  ac \otimes a \rangle\\
\end{pmatrix}
\eea
by the translation by $ab$.  The lower right block is also identified by this translation to the matrix
$$
\begin{pmatrix}
\langle c \otimes c | (U_c-z)^{-1} ca \otimes b \rangle&
\langle c \otimes c | (U_c-z)^{-1} cb \otimes c \rangle \\
\langle c \otimes a | (U_c-z)^{-1} ca \otimes b \rangle&
\langle c \otimes a | (U_c-z)^{-1} cb \otimes c \rangle 
\end{pmatrix}
$$ 
which, in turn, is identified to $P_0^+ (U_a - z)^{-1} P_1^-$ thanks to Proposition \ref{symcyc}. The non-diagonal term in the block $P_0^+ (U_a - z)^{-1} P_1^-$ are decoupled by the coin matrix $C_\pi $ on the vertex $a$. Finally, $\langle a \otimes b | (U_a - z)^{-1} ac\otimes a\rangle$ is equal to $\langle a \otimes a | G^z  ab \otimes c \rangle$, thanks to the last part of Propostion \ref{symcyc}. 
\ep

Hence, (\ref{matriximplict}) yields the sought for implicit equation for the restricted resolvent $P_lG^z P_l$:
\begin{cor}\label{matimp} For all $|z|<1$ such that $(\un+ P_3^- S_\pi P^+_2 G^z_\pi P_3^-)|_{P_3^-  \mathcal{K}_B^a}$ is invertible,
$P_lG^z P_l$ satisfies 
\be
P_lG^zP_l =  P_lG^zP_l \times A(P_l G^z P_l) + P_lG^z_\pi P_l ,
\ee
slightly abusing notations. 
\end{cor}
We now compute the operator $A(P_l G^z P_l)$ explictely. This will  allow us to derive an implicit equation for the Carath\'eodory  function associated with the spectral measure  $d\mu_{a\otimes a}$. 

Let us denote by $G_{1,2}^z = \langle a \otimes a | G^z | ab \otimes c \rangle$  the entry (1,2) in the $6 \times 6$ matrix $P_l G^z P_l$ we want to compute.
We use the parametrization (\ref{para+1}) for $C$ and start with the invertibility condition in Corollary \ref{matimp}. All computations below assume $z\in\D$ and all  matrices are expressed in the  bases  ordered according to definition (\ref{defproj}). According to this convention,  Proposition \ref{a=ag} reads
\be
P_2^+ G^z_\pi P_3^- = \langle a \otimes a | G^z   ab \otimes c \rangle \un.
\ee

Explicit computations yield the following expressions for the first part of (\ref{a=ag}):
\begin{lem} The $4\times 4$ matrix $(\un + P_3^- S_\pi P^+_2 G^z_\pi P_3^-)$ is inversible whenever 
\be
M(G_{1,2}^z) := (\beta^2 -\alpha \gamma)(G^z_{1,2})^2+2\beta G^z_{1,2} +1 \neq 0
\ee
and is bloc diagonal with identical $2\times 2$ blocs given by
\be
\frac{1}{M(G_{1,2}^z)} \begin{pmatrix} \beta G_{1,2}^z +1  & -\gamma G_{1,2}^z \\
-\alpha G_{1,2}^z & \beta G_{1,2}^z +1 \end{pmatrix}.\vspace{-.5
cm}
\ee 
Consequently, 
\be
P_1^+S_\pi P^+_2G_\pi P_3^- (\un + P_3^- S_\pi P^+_2 G_\pi P_3^-)^{-1} P_3^- S_\pi P_2^- = \begin{pmatrix}
K(G_{1,2}^z) & 0\\
0 &K(G_{1,2}^z) 
\end{pmatrix}
\ee
with $K(G_{1,2}^z) = 2\alpha \gamma G^z_{1,2} + (2\alpha \beta\gamma -\gamma^3 -\alpha^3) (G^z_{1,2})^2$.
\end{lem}
In turn, one checks that this implies that the matrix $A(P_lG^zP_l)$ in (\ref{matimp}) reads
\begin{prop} Assume $z\in \D$ and  $M(G_{1,2}^z)\neq 0$ and set 
$
D(G_{1,2}^z) = \beta - \frac{K(G^z)}{M(G_{1,2}^z)}. 
$
Then, 
\bea
&&A (P_lG^zP_l)= -\frac{1}{1-z^6} \times \\ \nonumber
&&\hspace{-1cm} \begin{pmatrix}
\beta+\alpha z^2+\gamma z^4 & \beta z+\alpha z^3+\gamma z^5 &\beta z^2+\alpha z^4+\gamma &\beta z^3+\alpha z^5+\gamma z&\beta z^4+\alpha +\gamma z^2&\beta z^5+\alpha z+\gamma z^3 \\
0 & 0 & 0 & 0 & 0 & 0 \\
\alpha+\gamma z^2+\beta z^4 & \alpha z+\gamma z^3+\beta z^5 &\alpha z^2+\gamma z^4+\beta &\alpha z^3+\gamma z^5+\beta z&\alpha z^4+\gamma +\beta z^2&\alpha z^5+\gamma z+\beta z^3 \\
D(G_{1,2}^z) z^3 & D(G_{1,2}^z) z^4 & D(G_{1,2}^z) z^5 & D(G_{1,2}^z) & D(G_{1,2}^z) z & D(G_{1,2}^z) z^2 \\
\gamma+\beta z^2+\alpha z^4 & \gamma z+\beta z^3+\alpha z^5 &\gamma z^2+\beta z^4+\alpha &\gamma z^3+\beta z^5+\alpha z&\gamma z^4+\beta +\alpha z^2&\gamma z^5+\beta z+\alpha z^3 \\
D(G_{1,2}^z) z & D(G_{1,2}^z) z^2 &D(G_{1,2}^z) z^3 & D(G_{1,2}^z) z^4 & D(G_{1,2}^z) z^5 & D(G_{1,2}^z) 
 \end{pmatrix}.
\eea
\end{prop}

From this expression, we relate entries of $P_lG^zP_l$ to $G_{1,2}^z$ to derive an equation for $G_{1,2}^z$.

\begin{lem} For $0<|z|<1$ such that $M(G_{1,2}^z)\neq 0$,
\be
G^z_{1,1} = \frac1z ( G^z_{1,2} -1 ), \ \ G^z_{l,1} = \frac1z G^z_{l,2},\ \  l = 2,\dots,6. 
\ee
\end{lem}
{\bf Proof:} Setting $a(z)=-(1-z^6)A (P_lG^zP_l)$, we get
\begin{align*}
-(1-z^6) G^z_{l,1} &= \sum_k G^z_{l,k} a(z)_{k,1} -(1-z^6) G^z_{\pi,l,1} 
                 = \frac{1}{z} \sum_k G^z_{l,k} a(z)_{k,2} -(1-z^6) G^z_{\pi,l,1} \\
                 &= \frac{1}{z} \sum_k (G^z_{l,k} a(z)_{k,2} -(1-z^6) G^z_{\pi,l,2}) +(1-z^6) G^z_{\pi,l,2}/z -(1-z^6) G^z_{\pi,l,1} \\
                 &= -\frac{1}{z}(1-z^6) G^z_{l,2} +(1-z^6) G^z_{\pi,l,2}/z -(1-z^6) G^z_{\pi,l,1}.
\end{align*}
Thus,
$
G^z_{l,1} = \frac{1}{z} G^z_{l,2} - G^z_{\pi,l,2}/z + G^z_{\pi,l,1}
$
and replacing $G^z_\pi$ by its values proves the result. \ep\\
\begin{rem}
The link between $G^z_{1,1}$ and $G^z_{1,2}$ together with (\ref{rescara}) yields 
\be
F(z)=2G^z_{1,2}-1\equiv 2g(z)-1
\ee
where $F$ is the Carath\'eodory function of $d\mu_{a\otimes a}$, and 
\be
g(z)=\bra a\otimes a | U_a(C)(U_a(C)-z)^{-1} a\otimes a\ket
\ee
appears in Theorem \ref{maintech}. 
\end{rem}

Similar considerations on the other matrix elements immediately yield
\begin{lem} For all $z\in \D$  such that $M(G^z_{1,2})\neq 0$,  we have
\bea
&&G^z_{1,1} z=  G^z_{1,2} -1, \ \ 
G^z_{1,4}(1+D(G^z_{1,2})) =z G^z_{1,3}, \ \ G^z_{1,6}(1+D(G^z_{1,2})) = z G^z_{1,5},\\
 &&\begin{pmatrix} \alpha & \beta +1 & \gamma \\
\gamma & \alpha & \beta+1  \\ \beta+1   & \gamma & \alpha \end{pmatrix}
\begin{pmatrix}G^z_{1,1}\\ G^z_{1,5} \\ G^z_{1,3}\end{pmatrix}=\begin{pmatrix}zG^z_{1,4}\\ zG^z_{1,2} \\ zG^z_{1,6} \end{pmatrix}.\label{mateq}
\eea
\end{lem}
\begin{rem}
Since the resolvent $G^z$ is analytic in $\D$, there are only finitely many $z$ such that $M(G_{1,2}^z)=0$ or $D(G_{1,2}^z)+1=0$ in any compact set of $\D$.
\end{rem}
Expressing $G^z_{1,6}, G^z_{1,4}, G^z_{1,1}$ in terms of $G^z_{1,5}, G^z_{1,3}, G^z_{1,2}$ in (\ref{mateq}), one deduces an equation for $g(z)=G^z_{1,2}$ by elimination of $G^z_{1,5}, G^z_{1,3}$ via appropriate linear combinations of the resulting equations. The result, which  implies Theorem \ref{maintech} directly, reads as follows.
Let 
\bea
P(g)&=& (\alpha^3+\beta^3+\gamma^3-3\alpha \beta\gamma+\beta^2-\alpha\gamma)g^2+2(\beta(\beta+1)-\alpha\gamma)g+(\beta+1)
\eea
so that $D(g)+1={P(g)/M(g)}$.
\begin{prop} Let $z\in \D$ be such that $M(g(z))\neq 0$. Then  $g(z)$ satisfies $\Phi(g(z),z)=0$, where
\bea 
\Phi(g,z)&=&P(g)\Big\{\alpha [\alpha z^2+(g-1)((\beta+1)^2-\alpha\gamma)][((\beta+1)^2-\alpha\gamma)P(g)+z^2\alpha M(g)]\\ \nonumber
&&\hspace{1.5cm}+(\beta+1)[(\beta+1)z^2+(g-1)(\alpha^2-(\beta+1)\gamma)]\times \\ \nonumber
&&\hspace{5.8cm}\times[(\alpha^2-(\beta+1)\gamma)P(g)+z^2(\beta+1)M(g)]\Big\}
\\  \nonumber
&&-((z^2-\gamma)g+\gamma)[(\alpha^2-(\beta+1)\gamma)P(g)+z^2(\beta+1)M(g)]\times \\
&&\hspace{6cm}\times [((\beta+1)^2-\alpha\gamma)P(g)+z^2\alpha M(g)]. \nonumber
\eea
\end{prop}
We spell out the detail, for the record.  
\begin{prop}\label{opc} We have
$\Phi(g,z)=\sum_{j=0}^5g^jc_j(z^2)$ with
\bea
\hspace{-1cm}c_5(x)&=& -x^3\{(\beta^2-\alpha \gamma)^2 \\
 &&+x^2\gamma(\beta^2-\alpha \gamma)(2 (\alpha^3+\beta^3+\gamma^3-3\alpha \beta\gamma)+3(\beta^2-\alpha\gamma)) \nonumber\\
&&+x(\alpha\beta+\alpha-\gamma^2)(\alpha^3+\beta^3+\gamma^3-3\alpha \beta\gamma+\beta^2-\alpha\gamma)\times\nonumber\\
&&\hspace{5cm}\times
(\alpha^3+\beta^3+\gamma^3-3\alpha \beta\gamma+3(\beta^2-\alpha\gamma)) \nonumber\\
&&+(\alpha+\beta+\gamma+1)(\alpha^2+\beta^2+\gamma^2-\beta\gamma-\alpha\gamma-\alpha\beta-\alpha-\gamma+2\beta+1)\times\nonumber\\
&&\hspace{5cm}\times(\alpha^3+\beta^3+\gamma^3-3\alpha \beta\gamma+\beta^2-\alpha\gamma)^2 \}.\nonumber\\
c_4(x)&=& x^3 4\beta(\alpha \gamma-\beta^2) \\
         &&+x^2\gamma(4\gamma^3\beta+7\beta^4+4\alpha^3 \beta-18\alpha\beta^3\gamma-12\alpha \beta\gamma+3\gamma^2\alpha^2+12\beta^3) \nonumber\\
&&+x2(\gamma^2-\alpha(\beta+1))(\alpha\gamma^4-\gamma^3\beta^2-4\gamma^3\beta-3\alpha^2\gamma^2(\beta+1)+\alpha^4\gamma\nonumber\\
&&\hspace{2cm}
+6\alpha\beta\gamma+4\alpha\beta^3\gamma+18\alpha \beta^2\gamma-\beta^5-\alpha^3\beta^2-7\beta^4-6\beta^3-4\alpha^3\beta)) \nonumber\\
&&-(\alpha+\beta+\gamma+1)(\alpha^2+\beta^2+\gamma^2-\beta\gamma-\alpha\gamma-\alpha\beta-\alpha-\gamma+2\beta+1)\times\nonumber\\
&&\hspace{0cm}\times(\alpha^3+\beta^3+\gamma^3-3\alpha \beta\gamma-3\beta^2+3\alpha\gamma-4\beta)(\alpha^3+\beta^3+\gamma^3-3\alpha \beta\gamma+\beta^2-\alpha\gamma) .\nonumber
\\
c_3(x)&=& x^3 2(\alpha \gamma-3\beta^2) +x^2 2\gamma(\alpha^3+4\beta^3+\gamma^3-6\alpha \beta\gamma+9\beta^2-3\alpha\gamma) \nonumber\\
&&+x2(\gamma^2-\alpha(\beta+1))
(-2\gamma^3+\gamma^3\beta+3\alpha\gamma(1-\beta^2)+12\alpha \beta\gamma\nonumber\\
&&\hspace{4.3cm}
-9\beta^2+\alpha^3\beta+\beta^4-8\beta^3-2\alpha^3) \\
&&+2(\alpha+\beta+\gamma+1)(\alpha^2+\beta^2+\gamma^2-\beta\gamma-\alpha\gamma-\alpha\beta-\alpha-\gamma+2\beta+1)\times\nonumber\\
&&\hspace{0cm}\times(\alpha^3+4\beta^3+\gamma^3-\alpha \gamma+3\beta^2-6\alpha\beta\gamma-\beta^4-\alpha^3\beta+3\alpha\beta^2\gamma-\beta\gamma^3-2\beta^5\nonumber\\
&&\hspace{3cm}-2\alpha^3\beta^2+8\alpha\beta^3\gamma+2\alpha^4 \gamma-6\alpha^2\beta\gamma^2-2\gamma^3\beta^2+2\alpha \gamma^4) .\nonumber\\
c_2(x)&=& -x^3 4\beta +x^2 2\gamma(6\beta+\beta^2-\alpha\gamma) \\
&&+x 2(\gamma^2-\alpha(\beta+1))
(\alpha^3+4\beta^3+\gamma^3+2\alpha\gamma-6\alpha \beta\gamma
-2\beta^2-6\beta) \nonumber\\
&&-2(\alpha+\beta+\gamma+1)(\alpha^2+\beta^2+\gamma^2-\beta\gamma-\alpha\gamma-\alpha\beta-\alpha-\gamma+2\beta+1)\times\nonumber\\
&&\hspace{0cm}\times(\alpha^3+4\beta^3+\gamma^3+\alpha \gamma-\beta^2-2\beta-6\alpha\beta\gamma-7\alpha\beta^2\gamma+3\beta^4+\alpha^3\beta+\beta\gamma^3+2\alpha^2\gamma^2) .\nonumber\\
c_1(x)&=& -x^3 +x^2 \gamma(3-2\beta) -x (\gamma^2-\alpha(\beta+1))
(3+4\alpha\gamma-4\beta-7\beta^2) \\
&&+(1+\beta)(\alpha+\beta+\gamma+1)
(1+4\alpha\gamma-3\beta-4\beta^2)\times
\nonumber\\
&&\hspace{3.5cm}\times
(\alpha^2+\beta^2+\gamma^2-\beta\gamma-\alpha\gamma-\alpha\beta-\alpha-\gamma+2\beta+1).\nonumber \\
c_0(x)&=& -x^2 \gamma+x 2(\beta+1) (\gamma^2-\alpha(\beta+1)) +(1+\beta)^2(\alpha+\beta+\gamma+1)
\times \\ \nonumber
&&\hspace{3.2cm}\times
(\alpha^2+\beta^2+\gamma^2-\beta\gamma-\alpha\gamma-\alpha\beta-\alpha-\gamma+2\beta+1).
\eea
\end{prop}
These formulae simplify substantially in case $C\in O(3)\cap \mbox{Circ}(3)$, by repeated application of the following identities that hold in $\mbox{CO}_+(3)$, see (\ref{param12})
\bea
&&\alpha+\beta+\gamma=0, \ \ \ 
\alpha\gamma-\beta(\beta+1)=0, \ \ \ 
\alpha^2+\beta^2+\gamma^2+2\beta=0
\nonumber\\&&
\gamma^2-\alpha(\beta+1)-\gamma=0, \ \ \ 
\alpha^3+\beta^3+\gamma^3-3\alpha \beta\gamma=0.
\eea
Altogether, this yields 
\begin{prop}\label{cjcop}
If  $C\in \mbox{ CO}_+(3)$, we have
$M(g)=-\beta g^2+2\beta g+1, $
and
\bea
c_5(x)&=&-\beta^2(x-1)(x^2+x(1-3\gamma)+1)\nonumber\\
c_4(x)&=&\beta^2(4x^3-9\gamma x^2+6\gamma x-1)\nonumber\\
c_3(x)&=&2\beta((1-2\beta)x^3+3\gamma(\beta-1) x^2+3\gamma x-(\beta+1))\nonumber\\
c_2(x)&=&-2\beta(2x^3-5\gamma x^2+\gamma (3\beta+4)x-(\beta+1))\nonumber\\
c_1(x)&=&-x^3+\gamma(3-2\beta) x^2-3\gamma (1-\beta^2)x+(\beta+1)^2\nonumber\\
c_0(x)&=&-\gamma x^2+2\gamma (\beta+1)x-(\beta+1)^2.
\eea
\end{prop}
The implicit equation satisfied by $g(z)$,
\be\label{spg}
\Phi(g(z),z^2)=\sum_{j=0}^5g^j(z)c_j(z^2)\equiv 0, \ \ \ z\in\D,
\ee
yields enough information to rule out point spectrum in the orthogonal case. 
\begin{prop}\label{eop}
For $C\in\mbox{Circ}(3)\cap O(3)\setminus\{\pm C_\sigma, \pm C_\pi\}$, 
\be
\sigma_{pp}(U_a(C))=\emptyset, \ \ \  \mbox{and} \ \ \ \sigma(U_a(C))=\overline{\sigma(U_a(C))}=-\sigma(U_a(C)).
\ee
\end{prop}
{\bf Proof:} Consider the symmetry first. Since the polynomials $c_j$ have real valued coefficients, they satisfy $c_j(z^2)=\overline{c_j(\bar{z}^2)}$. Hence $g(z)\equiv \overline{g(\bar{z})}$, $\forall z\in \D$, so that $\Re g(re^{i\theta})=\Re g(re^{-i\theta})$. This yields the supplementary symmetry $d\mu(\theta)=d\mu(-\theta)$, in the limit $r\ra 1^-$.

Consider now the point spectrum for $C\in\mbox{CO}_+(3)\setminus\{C_\sigma, C_\pi \}$.  Recall that $g(z)=(F(z)+1)/2$ is analytic in $\D$ and 
such that the weight of a point of the spectral measure $d\mu$ located at $\theta_0\in \T$ is given by 
\be
\mu(\{\theta_0\})=\lim_{r\ra1^-}(1-r)g(re^{i\theta_0})>0,
\ee 
according to (\ref{atom}).
Assume that $\theta_0\in \mbox{supp }d\mu_s\in \T$, {\it i.e.} $\Re g(re^{i\theta_0})\ra \infty $ as $r\ra 1^-$. This implies that $c_5(e^{2i\theta_0})=0$ and that $M(g(re^{i\theta_0}))$ 
is bounded away from zero for $r$ close to 1. For $\theta_0\in \T$ fixed and all $z\in \D$, let 
\be
\mu_{\theta_0}(z)=e^{-i\theta_0}g(z)(e^{i\theta_0}-z), \ \ \mbox{s.t.}\ \ \lim_{r\ra 1^-}\mu_{\theta_0}(re^{i\theta_0})=\mu(\{\theta_0\}).
\ee
Thus, (\ref{spg}) reads
\be\label{eqmu}
\sum_{j=0}^5\mu^j_\theta(z)\frac{e^{ij\theta_0}c_j(z^2)}{(e^{i\theta_0}-z)^j}\equiv 0, \ \ \ z\in\D.
\ee
First observe that $\pm 1$ is always a simple root of $c_5(z^2)$ and that $c_5(z^2)$ admits two pairs of conjugated simple zeros $\pm z_0(\gamma), \pm \bar{z_0}(\gamma)$ on ${\mathbb S}$ for $-1/3<\gamma <1$ only. Moreover, $c_4(z^2)\neq 0$ at these simple roots on ${\mathbb S}$. Then, for $\gamma=-1/3$, $c_5(z^2)$ admits double zeros at $\pm i$ which are simultaneously simple zeros of $c_4(z^2)$ and $c_3(z^2)$. The case $\gamma=1$ is excluded by assumption.
Hence two cases occur:
\\
i)  $e^{i\theta_0}$ is a simple zero of $c_5(z^2)$, so that (\ref{eqmu}) is equivalent to
\be\label{simple}
-\mu^5_{\theta_0}(z)e^{2i\theta_0}c'_5(e^{2i\theta_0})(1+O(z-e^{i\theta_0}))+\mu^4_{\theta_0}(z)c_4(e^{2i\theta_0})(1+O(z-e^{i\theta_0}))+O(z-e^{i\theta_0})\equiv 0,
\ee
ii) $e^{i\theta_0}$ is a double zero of $c_5(z^2)$ and a simple zero of $c_4(z^2)$ and $c_3(z^2)$, so that (\ref{eqmu}) is equivalent to
\be\label{double}
\mu^5_{\theta_0}(z)e^{2i\theta_0}c''_5(e^{2i\theta_0})(1+O(z-e^{i\theta_0}))-\mu^4_{\theta_0}(z)c'_4(e^{2i\theta_0})(1+O(z-e^{i\theta_0}))+O(z-e^{i\theta_0})\equiv 0.
\ee
For $z=re^{i\theta_0}$, we get for the two cases  in the limit $r\ra 1^-$
\bea\label{negweight}
i) &\mu^4(\{\theta_0\})(\mu(\{\theta_0\})e^{2i\theta_0}c'_5(e^{2i\theta_0})-c_4(e^{2i\theta_0}))=0\nonumber \\
ii) &\mu^4(\{\theta_0\})(\mu(\{\theta_0\})e^{2i\theta_0}c''_5(e^{2i\theta_0})-c'_4(e^{2i\theta_0}))=0.
\eea 
It is a matter of computation to check that the non-zero solutions to ({\ref{negweight}}) for the different values $e^{i\theta_0}$ of interest are all strictly negative, which implies that $\mu(\{\theta_0\})=0$. 

The case $C\in\mbox{CO}_-(3)\setminus\{- C_\sigma, - C_\pi\}$ is dealt with analogously, making use of Remark \ref{copcom}. The change of variables $\alpha\mapsto -\alpha, \beta\mapsto -\beta-2, \gamma\mapsto -\gamma$ in Proposition \ref{cjcop} yields the corresponding polynomial $c_j(z^2)$ for this case.  These polynomials have the same properties as in the previous case, so that the same conclusions hold.  \ep\\

\end{document}